%% file: main.tex
\documentclass[11pt]{article} %

\usepackage[backend=bibtex,style=numeric-comp,natbib=true,sortcites=true,
            doi=false,isbn=false,url=false,eprint=false]{biblatex} %

\usepackage[margin=1in]{geometry}

\usepackage{microtype}          %
\usepackage[T1]{fontenc}
\usepackage[full]{textcomp}

\usepackage{article} %
\usepackage{math} %
\usepackage{algo} %
\usepackage{wrapfig} %
\usepackage{placeins} %
\usepackage{graphicx} %
\usepackage{layout} %
\usepackage{setspace} %
\usepackage{mathtools} %
\usepackage{grffile} %
\usepackage{pdfpages} %
\usepackage{mdframed} %
\usepackage{flows}    %
\usepackage{booktabs}
\usepackage{tabulary} %
\usepackage{cellspace}          %
\usepackage{xcolor}             %
\usepackage{titlesec}           %

\newcommand{\defterm}[1]{{\emph{#1}}}

\makeatletter
\g@addto@macro\bfseries{\boldmath}
\makeatother

\bibliography{references}

\newcommand{\graybar}{{\color{Gray}\hrulefill}}

\newcommand{\mystrut}{\vrule height 1.25em depth .75em width 0pt }
\newcommand{\bigstrut}{\vrule height 2em depth 1.25em width 0pt }
\newcommand{\emstrut}[2]{\vrule height #1em depth #2em width 0pt }

\titlespacing*{\paragraph}{%
  0pt}{
  {\medskipamount}}{
  1em}

\begin{document}

\title{Faster~Algorithms for~Rooted~Connectivity in~Directed~Graphs}

\author{
Chandra Chekuri\thanks{Dept.\ of Computer Science, Univ.\ of Illinois, Urbana-Champaign, Urbana,
  IL 61801. {\tt chekuri@illinois.edu}. Supported in part by NSF grants
CCF-1910149 and CCF-1907937.}
\and
Kent Quanrud\thanks{Dept.\ of Computer Science, Purdue University,
  West Lafayette, IN 47909. {\tt krq@purdue.edu}.}
}

\maketitle

\nocite{gabow}

\begin{abstract}
  We consider the fundamental problems of determining the rooted and
  global edge and vertex connectivities (and computing the
  corresponding cuts) in \emph{directed} graphs.  For rooted (and
  hence also global) edge connectivity with small integer capacities
  we give a new randomized Monte Carlo algorithm that runs in time
  \begin{math}
    \apxO{n^2}.
  \end{math}
  For rooted edge connectivity this is the first algorithm to improve
  on the $\Omega(n^3)$ time bound in the dense-graph high-connectivity
  regime. Our result relies on a simple combination of sampling
  coupled with sparsification that appears new, and could lead to
  further tradeoffs for directed graph connectivity problems.

  We extend the edge connectivity ideas to rooted and global vertex
  connectivity in directed graphs. We obtain a
  $(1+\eps)$-approximation for rooted vertex connectivity in
  $\apxO{n W / \eps}$ time where $W$ is the total vertex weight
  (assuming integral vertex weights); in particular this yields an
  $\apxO{n^2/\eps}$ time randomized algorithm for unweighted graphs.
  This translates to a $\apxO{\vc n W}$ time exact algorithm where
  $\vc$ is the rooted connectivity. We build on this to obtain
  similar %
  bounds for global vertex connectivity.

  Our results complement the known results for these problems in the
  low connectivity regime due to work of Gabow \cite{gabow} for edge
  connectivity from 1991, and the very recent work of Nanongkai et
  al.\ \cite{nsy} and Forster et al.\ \cite{fnsyy} for vertex
  connectivity.
\end{abstract}

\clearpage

\input{intro}

\input{edge-connectivity}

\input{vertex-connectivity}

\printbibliography

\end{document}

%% file: intro.tex
\section{Introduction}
\labelsection{intro}

Let $G = (V,E)$ be a simple directed graph; that is, a directed graph
with no parallel edges. Recall that $G$ is \defterm{strongly
  connected} if there is a path from any vertex $a \in V$ to any
vertex $b \in V$. The \defterm{edge connectivity} is the minimum
number of edges that need to be removed so that $G$ is \emph{not}
strongly connected. The corresponding set of edges is called the
\defterm{minimum edge cut}. The \defterm{vertex connectivity} is the
minimum number of vertices that need to be removed so that the
remaining graph is not strongly connected or has only one vertex.  The
corresponding set of vertices is called the \defterm{minimum vertex
  cut}.  These problems generalize to weighted settings where the
edges and vertices are assigned positive weights and the goal is to
find the minimum weight edge or vertex cut. Determining the edge and
vertex connectivities and finding the corresponding minimum cuts are
among the basic problems in graph algorithms.  This work obtains
faster randomized algorithms for finding minimum edge and vertex cuts
in directed graphs, especially in the dense setting. The algorithms
are based on a simple technique which could be of independent
interest.

Our interest is actually in the more general \defterm{rooted
  connectivity} problems. Let $r \in V$ be a fixed vertex, called the
\defterm{root}. The \defterm{$r$-rooted edge connectivity} is the
minimum number of edges that have to be removed so that there is some
vertex in $V -r$ that $r$ cannot reach.  An algorithm for rooted edge
connectivity easily implies an algorithm for edge connectivity, by
fixing any root and returning the minimum of the rooted connectivity
in $G$ and the rooted connectivity in the graph obtained by reversing
all the edges in $G$.  Another important motivation for investigating
rooted connectivity is the fundamental result by Edmonds
\cite{edmonds-70} that the $r$-rooted edge connectivity equals the
maximum number of edge-disjoint arboresences rooted at $r$. We refer
the reader to \cite{schrijver,frank} for further connections in
combinatorial optimization. Similarly, the \defterm{$r$-rooted vertex
  connectivity} is the minimum number of vertices (excluding $r$) that
have to be removed so that $r$ cannot reach some vertex in the
residual graph. Algorithms for rooted vertex connectivity also lead to
global vertex connectivity by a similar but somewhat more involved
reduction.  There is a long and rich history of developing algorithms
for determining the edge and vertex connectivity.  We first note that
all of these connectivity and cut problems reduce to a polynomial
number of $(s,t)$-cut and flow problems by standard reductions.
Beyond $(s,t)$-flow, an interesting algorithmic landscape emerges with
different running times depending on whether we are interested in edge
or vertex cuts, directed or undirected graphs, and weighted or
unweighted graphs.

\paragraph{Rooted and global edge-connectivity:}
We first discuss edge connectivity in directed graphs.  Let $\ec$
denote either the rooted or global edge connectivity of the graph
depending on the context. One can compute both via $O(n)$
$(s,t)$-minimum cut computations. For the simple and unweighted
directed graph setting, Mansour and Schieber \cite{mansour-schieber}
improved on this and gave algorithms that run in $\bigO{m n}$ time or
in $\bigO{\ec^2 n^2}$ time for global connectivity. It was also
observed by Alon (cf.\ \cite{mansour-schieber}) that this approach can
be parameterized by the minimum out-degree $\delta^+$ to obtain a
$\bigO{n \log{\delta^+} \ectime{m,n} / \delta^+}$ algorithm, where
$\ectime{m,n}$ denotes the running time for $(s,t)$-edge
connectivity\footnote{Depending on the context, we let $\ectime{m,n}$
  denote the running time for $(s,t)$-cut in either a simple or a
  weighted directed graph with $m$ edges and $n$ vertices.}.  Gabow
\cite{gabow} then gave a $\bigO{m \ec \log{n^2 / m}}$ for rooted
connectivity in graphs with integer capacities.  Gabow's algorithm is
based on Edmonds' theorem described above.  Gabow's algorithm is
nearly linear time for sparse unweighted graphs, and remains the
fastest algorithm for small $\ec$ for both rooted and global edge
connectivity. It is interesting that Gabow's algorithm is not based on
$(s,t)$-flow. For directed graphs with arbitrary edge capacities, Hao
and Orlin \cite{hao-orlin} gave an $\bigO{m n \log{n^2 / m}}$
algorithm for rooted connectivity by adapting the push-relabel max
flow algorithm; in fact their algorithm computes the
$(r,v)$-connectivity for all $v \in V-r$.  Thereafter there have been
no direct running time improvements to rooted or global edge
connectivity in directed graphs but we point out that there have been
numerous breakthroughs in the running times for $(s,t)$-flow and
connectivity
\cite{goldberg-rao,orlin-13,lee-sidford,madry-13,madry-16,liu-sidford-20,kls-20,brand++,glp-21}. In
particular, starting with the work of Goldberg and Rao
\cite{goldberg-rao}, the running time for $(s,t)$-flow is $o(mn)$
which breaks the flow-decomposition barrier. Motivated by these
developments and several others, there has been a resurgence of
interest and literature on faster graph algorithms for several
fundamental problems. Despite these developments there has been no
algorithm for rooted edge-connectivity in simple directed graphs that
is faster than $O(n^3)$ in the worst case. In this paper we obtain a
nearly quadratic time algorithm which also applies to graphs with
small integer capacities.\footnote{Here and throughout $\apxO{\cdots}$
  hides polylogarithmic factors in $m$ and $n$.  We note that the
  ideas introduced in this work are simple and the logarithmic factors
  they generate are easy to account for. However \reftheorem{main}
  also uses the recent $(s,t)$-flow algorithm of \cite{brand++} with
  running time $\ectime{m}{n} = \apxO{m + n^{1.5}}$, which has large
  logarithmic factors.}

\begin{restatable}{theorem}{QuadraticEdgeConnectivity}
  \labeltheorem{main} Let $G = (V,E)$ be a simple directed graph with
  $m$ edges $n$, vertices, and integer edge weights $w: E \to [U]$.
  Then the minimum rooted $r$-cut can be computed with high
  probability in $\apxO{n^2 U}$ randomized time.
\end{restatable}

\input{ec-running-times}

This running time is particularly compelling when the rooted edge
connectivity $\ec$ is high.

\paragraph{Rooted and global vertex-connectivity:}
We now consider (rooted) vertex connectivity in directed graphs. It is
well known that for fast algorithms, global vertex connectivity is
more involved than edge connectivity and the running times are more
varied. While the rooted vertex connectivity can be reduced to
computing $O(n)$ $(s,t)$-cuts, the global version, if done naively,
would require $\Omega(n^2)$ calls to the $(s,t)$-cut problem since it
is not obvious how to find a vertex that is not part of the minimum
global vertex cut.  There is a large body of literature and we
highlight the leading (randomized) running times, where we state
running times for randomized algorithms with high probability of
success. Let $\vc$ denote the weight minimum vertex cut, where we
assume the minimum weight of any vertex is 1. For large $\vc$ and
general capacities, there is a randomized algorithm by Henzinger \etal
\cite{ghr} (extending the directed edge connectivity algorithm of
\cite{hao-orlin}) that runs in $\bigO{m n \log{n}}$ time. For small
values of $\vc$ in the unweighted setting, recent randomized
algorithms by Forster \etal \cite{fnsyy} based on \emph{local
  connectivity} have obtained $\apxO{m \vc^2}$ and
$\apxO{n \vc^3 + \vc^{3/2} \sqrt{m} n}$ running times. For more
intermediate values of $\vc$, there are also randomized
$\apxO{\vc m^{2/3} n}$ and $\apxO{\vc m^{4/3}}$ time algorithms
\cite{nsy} as well as an $\bigO{n^{\omega} + n \vc^{\omega}}$ time
algorithm \cite{cheriyan-reif}, where $\omega \approx 2.3728596$ is
the current exponent for fast matrix multiplication
\cite{alman-vw}. There is also recent interest in obtaining fast
$\epsmore$-approximation algorithms for minimum vertex cut
\cite{nsy,fnsyy}. In particular \cite{fnsyy} obtains a randomized
algorithm with running time $\apxO{m \vc / \eps}$. Here too we can ask
whether one can obtain algorithms that beat $n^3$ in the worst-case
for rooted and global vertex connectivity in directed graphs, even
when allowing for a constant factor approximation. We obtain the
following theorem.

\begin{restatable}{theorem}{SimpleVCRooted}
  \labeltheorem{simple-vc-rooted} \labeltheorem{rooted-vc-reduction}
  \setupVGintr Let $\vc$ be the rooted vertex connectivity from
  $r$. Let $W = \sum_{v \in V} \weight{v}$ be the total weight of the
  graph.  For any $\eps > 0$ a $(1+\eps)$-approximate rooted minimum
  vertex cut can be computed with high probability in
  $\apxO{m + n (W - \vc) / \eps}$ randomized time; for unit weights
  this is $\apxO{m + n (n - \vc) /\eps}$. The rooted connectivity can
  be computed with high probability in $\apxO{m + \vc n (W - \vc)}$ time.
\end{restatable}
Note that $W \geq n$ in the above running times.  We point out that
the approximation algorithm's running time is independent of
$\vc$. This large $\vc$ regime has been challenging for previous
approaches. The rooted connectivity algorithm, when combined with
sampling and other ideas, gives the following theorem for global
vertex connectivity.  As we remarked, the reduction from global to
rooted is not as clean for vertex connectivity as it is for edge
connectivity.

\begin{restatable}{corollary}{SimpleVCGlobal}
  \labelcorollary{simple-vc-global} \setupVGint Let
  $W = \sum_{v \in V} \weight{v}$ be the total vertex weight of the
  graph.  Let $\vc$ be the global vertex connectivity of $G$. There is
  a randomized algorithm that for any $\eps > 0$ outputs a
  $(1+\eps)$-approximate minimum vertex cut with high probability in
  time $\apxO{nW / \eps}$.  There is a $\apxO{\vc n W}$ time
  randomized algorithm that computes the (exact) minimum vertex cut
  with high probability. In particular, for unit weights, the running
  time is $\apxO{\vc n^2}$.
\end{restatable}

\input{vc-running-times}

%
%

\subsection{Key ideas}
Our algorithms are based on a simple but key idea that we briefly
outline below. We focus on the edge-connectivity case since the idea
for vertex connectivity is essentially the same with some
modifications.  We would like to take advantage of recent developments
on fast algorithms for $(s,t)$-cut and reduce to solving a small
number of such cut problems in a black box fashion (unlike the
approach of \cite{hao-orlin} based on the properties of a specific
flow algorithm). For undirected graph global connectivity there has
been very recent exciting progress by Li and Panigrahi
\cite{li-panigrahi} reducing to a \emph{logarithmic} number of
$(s,t)$-cuts. However, the technique makes strong use of the symmetry
of the edge-cut function which are absent in the directed graph
setting. In a different direction the work of Nanongkai, Saranurak,
and Yingchareonthawornchai \cite{nsy-19} and follow up improvements by
Forster \etal \cite{fnsyy}, developed fast algorithms for global
connectivity based on \emph{local connectivity} and
\emph{randomization}. At a high-level they use sampling to identify
two vertices $s,t$ on the opposite sides of a cut and then reduce to
$(s,t)$-cut, or they use a local-connectivity algorithm from each
vertex $v \in V$. This approach is particularly well-suited for
\emph{small} connectivity.

For directed graph edge connectivity Gabow's algorithm with running
time $\apxO{m \ec}$ is very good. In order to beat $O(n^3)$ in the
worst case, the bottleneck is the dense graph regime with high
connectivity. We have two main ideas that are particularly well suited
to this regime. First, we focus on the rooted case even though it may
appear to be more difficult than the global connectivity case. The
global connectivity can be much smaller than the rooted connectivity;
for instance the graph may not be strongly connected, in which case
the global edge connectivity is $0$, while the rooted connectivity for
a particular root can still be $\Omega(n)$. Consider rooted
connectivity from a given vertex $r$. In order to reduce to
$(s,t)$-cut we would like to find a node $t$ such that $t$ is the sink
side of a minimum $r$-cut. Let $T \subseteq V$ be a sink side of a
minimum $r$-cut and hence $\ec = \sizeof{\delta^{-}(T)}$; here
$\delta^-(T)$ denotes the set of edges entering $T$.  If $\sizeof{T}$
is large we can randomly sample a small number of vertices and we will
succeed with good probability in finding a vertex from $T$. Therefore
the difficult case is when $\sizeof{T}$ is small and this is the
setting in which we make our key observation: if the graph is
\emph{simple} (or edge capacities are small) and the sink side of a
minimum $r$-cut is small (but not a singleton!), then $T$ cannot have
a high-degree vertex. How can we take advantage of this? Since we are
working with the rooted problem, we can shrink all high-degree
vertices into the root $r$! In other words we can \emph{sparsify} the
graph if the sink side is small and compensate for the higher sampling
rate (and larger number of $(s,t)$-cut computations) we need to find a
vertex on the sink side.  Simple in retrospect, this tradeoff between
sparsification and sampling rate coupled with guessing the size of the
sink component gives us the overall algorithm with some additional
technical work.  We believe that our high-level idea will find use in
other contexts when combined with other techniques.

\paragraph{Recent related work.}
A recent and independent work of \cite{li+21} has obtained an
$\apxO{m n^{1-1/12 + o(1)}}$ time algorithm for vertex connectivity in
directed and unweighted graphs. We have not yet had time to digest and
make a proper comparison to \cite{li+21}.

Recent followup work by one of the authors has extended the ideas in
this work to obtain $\epsmore$-approximation algorithms for
\emph{weighted} graphs, for rooted and global, edge and vertex
connectivity, with $\apxO{n^2 / \eps^{\bigO{1}}}$ running times
\cite{quanrud-apx-cuts}.


%% file: ec-running-times.tex
\begin{table}[t]
  \everymath{\displaystyle}
  \renewcommand\tabularxcolumn[1]{m{#1}}%
  \centering
  \begin{tabularx}{\textwidth}{ c | >{\mystrut}c | X |}
    \cline{2-3}
    &
      \begin{math}
        \bigO{n \edgeconn{m,n}}
      \end{math}
      & Trivial. Also holds for rooted connectivity. \\
      \cline{2-3} &
      \begin{math}
        \bigO{n \edgeconn{m,n}_{\lambda}}
      \end{math}
    &
      Matula \cite{matula-87}. Also holds for rooted connectivity. \\
    \cline{2-3}
    &
      \begin{math}
        \bigO{m n},
      \end{math}
      \begin{math}
        \bigO{\lambda^2 n^2}
      \end{math}
    & Mansour and Schieber \cite{mansour-schieber} \\
    \cline{2-3}                      %
    *
    &
      \bigstrut
      \begin{math}
        \bigO{\frac{n \log \delta}{\delta} \edgeconn{m,n}}
      \end{math}
    & Alon (cf.\ Mansour and Schieber \cite{mansour-schieber}). $\delta$ is the
      minimum out-degree in the graph.
    \\
    \cline{2-3} * &
                    \begin{math}
                      \bigO{m \lambda \log{n^2 / m}}
                    \end{math}
    & Gabow \cite{gabow}. Also holds for rooted connectivity.
    \\
    \cline{2-3} * & \bigstrut
    \begin{math}
      \apxO{n^{2}}
    \end{math}
    & \reftheorem{main}. Randomized. Also holds for rooted
    connectivity.
    \\
    \cline{2-3}
  \end{tabularx}

  \caption{Running times for finding the minimum cut in unweighted
    directed graphs (i.e., $U = 1$). $\edgeconn{m,n}$ denotes the
    running time of computing $(s,t)$-connectivity (in unweighted
    graphs). See also
    \cite[\S15.3a]{schrijver}. \labeltable{simple-ec}}

  \graybar
\end{table}

%% file: vc-running-times.tex
\begin{table}[t]
  \everymath{\displaystyle}              %
  \renewcommand\tabularxcolumn[1]{m{#1}} %
  \centering
  \begin{tabularx}{\textwidth}{| >{\emstrut{1.35}{.75}}c | X |}
    \hline %
    $\bigO{n^2 \vctime{m,n}}$
    & Trivial. \\
    \hline %
    $\bigO{n \vctime{m,n} \log{n}}$ & Trivial.  Randomized. $\vc \leq .999n$ \\
    \hline %
    $\bigO{\vc n \vctime{m,n}}$ & Podderyugin \cite{podderyugin},
                                  Even and Tarjan \cite{even-tarjan}
    \\
    \hline %
    $\bigO{n^{\omega} + n \vc^{\omega}}$ & Cheriyan and Reif \cite{cheriyan-reif}.
    \\
    \hline %
    $\bigO{\vc m n}$, $\bigO{(\vc^3 + n) m}$ & Henzinger \etal \cite{ghr}.
    \\
    \hline %
    $\bigO{m n \log{n}}$ & Henzinger \etal \cite{ghr}. Randomized.
    \\
    \hline $\bigO{\min{\vc^{5/2}, \vc n^{3/4}} m + m n}$
    &
      Gabow \cite{gabow-06}.
    \\
    \hline $\apxO{\vc m^{2/3} n}$, $\apxO{\vc m^{4/3}}$
    &
      {Nanongkai \etal \cite{nsy}. Randomized.} %
    \\ \hline %
    $\apxO{m \vc^2}$, $\apxO{n \vc^3 + \vc^{3/2} m^{1/2} n}$
    &
      {Forster \etal \cite{fnsyy}. Randomized.} \\
    \hline %
    $\apxO{n^2 \vc}$ %
    & \refcorollary{simple-vc-global}. Randomized. %
    \\
    \hline %
  \end{tabularx}

  \caption{A table of running times for finding the minimum vertex cut
    in unweighted directed graphs (i.e., $W = n$). $\vctime{m,n}$
    denotes the running time of computing $(s,t)$-vertex connectivity,
    and is at most $\apxO{m + n^{1.5}}$ \cite{brand+}.  All randomized
    algorithms above are correct with high probability.  See also
    \cite[\S15.2a]{schrijver} and \cite{fnsyy}. \labeltable{simple-vc}}

  \graybar
\end{table}


%% file: edge-connectivity.tex
\section{Edge connectivity}
\labelsection{edge-connectivity}

\labelsection{edge-overview}

In this section, we prove the main theorem for edge connectivity,
\reftheorem{main}. To this end, we will first introduce the main key
lemma, called the \emph{Rooted Sparsification Lemma}, in
\refsection{overview-sparsification}.  In \refsection{bucket}, we give
a lemma that applies the Rooted Sparsification Lemma to give a faster
algorithm when the number of vertices in the sink component is known
to be in a fixed interval between $1$ and $n$. \reftheorem{main} is
then proven in \refsection{main-proof}, applying the ideas from
\refsection{bucket} to each of a small family of intervals.

\subsection{The Rooted Sparsification Lemma for Edge Connectivity}

\labelsection{overview-sparsification}
\labelsection{edge-sparsification}

We introduce the key technical ingredient that we call the Rooted
Sparsification Lemma.  This lemma says that if the sink component of
the minimum $r$-cut is small, then unless it is a singleton cut (which
is easy to find directly), we can contract all vertices with high
in-degree into the root while preserving the minimum rooted cut
exactly. The result is a smaller and sparser graph in which we can
find the minimum rooted cut faster. Later we will see that the running
time saved by operating on a smaller graph makes up for the difficulty
in identifying a vertex from a smaller sink component.

\begin{restatable}{lemma}{SimpleSparsification}
  \labellemma{simple-sparsification}
  \setupGwr Let
  $k \in \naturalnumbers$. Consider the graph $\bar{G}$ obtained by
  contracting all vertices with weighted in-degree $\geq (1 + U) k$
  into $r$. Let $\bar{r}$ denote the contracted vertex in
  $\bar{G}$. Then we have the following.
  \begin{enumerate}
  \item $\bar{G}$ is a multigraph with less than $(1 + U) n k$ edges.
  \item If the minimum number of vertices in a sink component of a
    minimum $r$-cut has greater than 1 and less than or equal to $k$
    vertices, then the minimum $r$-cut and the minimum $\bar{r}$-cut
    are the same.
  \end{enumerate}
\end{restatable}

\begin{figure}
  \centering

  \begin{minipage}{.5\paperwidth}
    \centering
    \includegraphics{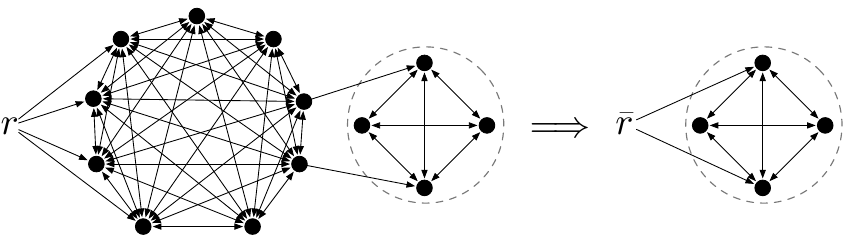}
    \caption{An example of the Rooted Sparsification Lemma in
      action. In particular, contracting the high in-degree vertices
      into $r$ leaves the sink component of the minimum $r$-cut
      intact.\labelfigure{rooted-sparsification}}
  \end{minipage}

  \bigskip

  \graybar
\end{figure}

Note that contraction cannot reduce the value of $r$-cut.  An example
illustrating the lemma is given in
\reffigure{rooted-sparsification}. The proof is in two steps.

\paragraph{Small sinks make small cuts (except for singletons).}
The first step towards the Rooted Sparsification Lemma for edge
connectivity is the following basic observation relating the
connectivity to the number of vertices in the sink component of a
minimum rooted cut. For simple graphs (i.e., $U = 1$), the following
lemma says that \emph{except for the case where the minimum rooted cut
  is achieved by a singleton}, the rooted connectivity is less than
the number of vertices in the sink component of the cut. With
capacities between $1$ and $U$, we obtain a similar inequality except
scaled by $U$. See \reffigure{small-sink-small-edge-cut} for an
illustration of the following lemma.

\begin{restatable}{lemma}{SimpleSink}
  \labellemma{small-sink-small-cut} \setupRC Let $k$ be the minimum
  number of vertices in a sink component of a minimum $r$-cut. Then
  either $k =1$ or $\lambda < U k$.
\end{restatable}
\begin{restatable}{proof}{SimpleSinkProof}
  Let $T$ be the set of vertices on the sink-side of a cut with
  $\lambda$ edges. Suppose $k = \sizeof{T} > 1$. Every vertex in $T$
  has weighted in-degree $> \lambda$. Consider all edges with head in $T$.
  Because $G$ has capacities
  between $1$ and $U$, of all
  the edges with head in $T$, at most $k(k-1) U$ total weight have
  their tail in $T$ as well. Thus
  \begin{math}
    \lambda > k \lambda - k(k-1) U.
  \end{math}
  Rearranging, we have
  \begin{math}
    k(k-1) U > (k-1) \lambda,
  \end{math}
  hence $k U > \lambda$.
\end{restatable}

\begin{remark}
  The above argument is simple and (unsurprisingly) we realized that a
  similar line of reasoning has appeared in previous work
  \cite{mansour-schieber} (though towards a different algorithmic
  approach and not in the context of rooted connectivity).
\end{remark}

\begin{figure}
  \centering

  \begin{minipage}{.5\paperwidth}
    \centering \includegraphics{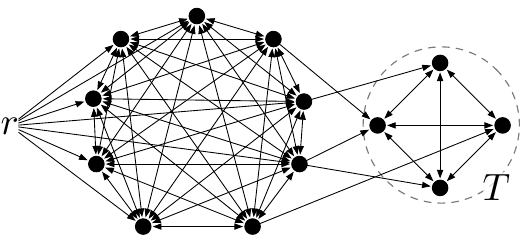}
    \caption{\small The set of vertices $T$ has 4 vertices and there
      are 5 edges crossing into $T$. \reflemma{small-sink-small-cut}
      implies that $T$ cannot be the sink component of the minimum
      $r$-cut. Indeed, there are singleton cuts of degree 4 inside
      $T$. \labelfigure{small-sink-small-edge-cut}}
  \end{minipage}

  \bigskip

  \graybar
\end{figure}

\paragraph{Small sinks are sparse sinks.}
We now prove the Rooted Sparsification Lemma,
\reflemma{simple-sparsification}. The high level argument is very
simple and we first give an informal argument to emphasize the
intuition. If the sink component of the minimum $r$-cut is small, then
by \reflemma{small-sink-small-cut}, the minimum $r$-cut is also
small. Suppose for the sake of discussion that the graph is simple
(i.e., $U = 1$). If both the minimum $r$-cut and the sink component
are small and the graph is simple, then every vertex in the sink
component has small in-degree. The contrapositive implies that every
high in-degree vertex is on the source side of the cut. Thus the high
in-degree vertices can be safely contracted into the root.

\begin{proof}[Proof of the rooted sparsification lemma]
  Recalling the statement of the lemma, it is easy to see that
  contracting all vertices with weighted in-degree $\ge (1+U)k$ into
  $r$ results in a multigraph $\bar{G}$ in which every vertex has
  weighted in-degree $< (1 + U) k$, and hence there are at most
  $(1 + U) n k$ edges total.

  Let $T$ be the sink component of a minimum $r$-cut.  Observe that
  contracting into $r$ cannot decrease the edge connectivity. If one
  can show that no vertices in $T$ are contracted into $\bar{r}$, then
  $T$ is the sink component of a minimum $\bar{r}$-cut as well.

  By \reflemma{small-sink-small-cut}, the minimum $r$-cut has size
  $\lambda < U k$. Because $G$ is simple and $T$ has $\leq k$
  vertices, every vertex in $T$ has in-degree less than
  $\lambda + k < (1 + U) k$. Thus any contracted vertex is outside
  of $T$. This completes the proof.
\end{proof}

\subsection{Rooted connectivity for a fixed range of component sizes}

\labelsection{bucket}

Applying the Rooted Sparsification Lemma usefully requires a fairly
tight upper bound on the number of vertices in the sink component of
the minimum $r$-cut. In this section, we assume we are given a lower
bound $k_1$ and an upper bound $k_2$ on the number of vertices in the
sink component, and develop algorithms for the minimum rooted cut in
this parametrized regime. The running times are decreasing in $k_1$
and increasing in $k_2$; that is, they are better for tighter bounds
on the number of vertices in the sink component.

\begin{lemma}
  \labellemma{known-component-size}
  \setupRC %
  Let $k_1, k_2 \in \naturalnumbers$ with $1 \leq k_1 \leq k_2 \leq n$. Suppose the
  sink component of the minimum $r$-cut has between $k_1$ and $k_2$
  vertices. Then the minimum $r$-cut can be computed with constant
  probability in
  \begin{align*}
    \bigO{m + \frac{n}{k_1} \parof{\edgeconn{\min{m,n k_2 U}}{n}}} \text{ time.}
  \end{align*}
\end{lemma}

\begin{proof}
  We first consider the case $k_1 > 1$.  By
  \reflemma{simple-sparsification}, we can reduce the number of edges
  to $m'= \bigO{k_2 n U}$ while preserving the $r$-cut and retaining
  all $k_1$ or more vertices in the sink-side component.  Let us
  sample $\bigO{n / k_1}$ sink vertices $t$ in the remaining graph
  uniformly at random, and compute the minimum $(r,t)$-cut for
  each. This takes
  \begin{math}
    \edgeconn{\min{m,m'}}{n} = \edgeconn{\min{m,k_2 n U}}{n}
  \end{math}
  time for each instance, as desired.  With constant probability, at
  least one sink is sampled out of the sink component of the minimum
  $r$-cut, which will return the minimum $r$-cut.

  If $k_1 = 1$, then we must also address the possibility of a
  singleton cut. We apply the above for $k_1 = 2$ and compare the
  output to all of the singleton $r$-cuts, and output the smallest of
  these cuts.
\end{proof}

\subsection{Rooted connectivity for small sink components}

\begin{lemma}
  \labellemma{ec-small-sink} \setupGwintr
  Let $k \in \naturalnumbers$ be a
  given parameter.
  There is a deterministic algorithm that runs in
  \begin{math}
    \bigO{m + n k^2 U^2 \log{\max{n/kU}}}
  \end{math}
  time and returns an $r$-cut with the following guarantee.  If the
  sink component of a minimum $r$-cut has at most $k$ vertices, then
  the algorithm will return a minimum $r$-cut.
\end{lemma}

\begin{proof}
  If the sink side of the minimum cut has less than $k$ vertices, then
  via \reflemma{small-sink-small-cut}, either a singleton induces a
  minimum $r$-cut, or the minimum $r$-cut has size $\lambda < U
  k$. For the latter case, we apply the rooted sparsification lemma
  and reduce the graph to $\bigO{n k U}$ edges while preserving the
  minimum $r$-cut. We apply Gabow's algorithm \cite{gabow} to the
  sparsified graph and it runs $\bigO{n k^2 U^2 \log{\max{1,n/kU}}}$
  time, and either finds a minimum rooted cut or certifies that the
  $r$-cut value in the sparsified graph has value $\ge kU$. We compare
  the output with all singleton $r$-cuts.
\end{proof}

\subsection{Algorithm for rooted edge connectivity}

\labelsection{main-proof} We now prove the main theorem for edge
connectivity, \reftheorem{main}.  By \reflemma{known-component-size},
if the number of vertices in the sink component is known, then we can
reduce very efficiently to $(s,t)$-connectivity by either sparsifying
the graph (if the number is small) or easily guessing a sink (if the
number is large). More generally, we can pursue both strategies
relative to any given upper and lower bounds on the number of vertices
in the sink component. Meanwhile, for small component sizes (that are
not singletons), we can still sparsify the graph, while the cut size
must be small, which combine to produce fast running times via
\cite{gabow} in \reflemma{ec-small-sink}. The only unknown is the
number of vertices in the sink component. Here we guess the number of
vertices up to a constant factor, which only requires enumerating a
logarithmic number of guesses. We restate \reftheorem{main} for the
sake of convenience.

\QuadraticEdgeConnectivity*

\begin{proof}
  Let $\ell \in [n]$ be a parameter to be determined.  The sink
  component of the minimum $r$-cut either (a) is a singleton, (b) has
  at most $\ell$ vertices, or (c) has between $2^i$ and $2^{i+1}$
  vertices for some $i \geq \logdown{\ell}$. For each of these categories we
  apply a subroutine and take the minimum of the cut values found.

  Singleton cuts are easy to evaluate in $O(m)$ time.  Let
  $i_0 = \logdown{\ell}$ and $i_1 = \max{\logup{m/nU}, i_0 + 1}$. For
  $i = i_0, \dots, i_1 - 1$, let $k_i = 2^i$. Let $k_{i_1} = n$.  For
  (b) we invoke \reflemma{ec-small-sink} with maximum sink component
  $k_{i_0}$.  To address (c), for each $i = i_0,\dots, i_1 - 1$, we
  invoke \reflemma{known-component-size} $\bigO{\log n}$ times with
  lower bound $k_i$ and upper bound $k_{i+1}$ on the number of
  vertices in the sink component. We use
  $\ectime{m}{n} = \apxO{m + n^{1.5}}$ \cite{brand++}.  The combined
  running time is
  \begin{math}
    \apxO{n^2 U + \frac{n^{2.5}}{\ell} + n \ell^2 U^2}.
  \end{math}
  For $\ell = \sqrt{n} / U$, this gives the claimed running time.
\end{proof}


%% file: vertex-connectivity.tex
\section{Rooted and global vertex connectivity}

In this section, we describe and analyze the approximation algorithms
for rooted and global vertex connectivity. The high-level approach is
similar to the previously discussed algorithm for edge connectivity.
The first step, \reflemma{vertex-sparsification}, is a variant of the
Rooted Sparsification Lemma that applies to (approximate) vertex
connectivity. It plays a similar role as its counterpart for edge
connectivity, allowing one to sparsify the graph when the sink
component of the minimum rooted vertex cut is small. The proof of
\reflemma{vertex-sparsification} is given
\refsection{vertex-sparsification}.  We then give an algorithm
specific to (roughly) the number of vertices in the sink component in
\refsection{parametrized-vertex-cut}. We use this algorithm as a
subroutine in the final algorithm for approximate rooted connectivity
in \refsection{rooted-vc}.  In
\refsection{global-vertex-connectivity}, we give the reduction from
approximate global vertex connectivity to approximate rooted vertex
connectivity.  The exact global vertex connectivity algorithm for
integer weights follows from an appropriate choice of error parameter.

\subsection{Rooted sparsification for approximate vertex
  connectivity}

\labelsection{vertex-sparsification}

Recall that a key idea in the algorithm for (rooted) edge connectivity
was the Rooted Sparsification Lemma, which allows us to substantially
decrease the number of edges when the sink component of the minimum
rooted cut is small.  Underlying the rooted sparsification lemma for
edge connectivity was a direct relation between the size of the sink
component and the weight of the minimum edge cut ---
\reflemma{small-sink-small-cut} in
\refsection{edge-sparsification}. But this relation does not hold for
vertex connectivity, even in unweighted and undirected graphs -- even
if the sink component is small, the vertex in-cut can be very large.
For example, for arbitrarily large $n$ and any fixed constant $k$, let
$S = K_n$ be a clique of size $n$ and let $T = K_k$ be a clique of
size $k$. Add edges between all $s \in S$ and all $t \in T$. Let $r$
be an additional root vertex connected to every vertex in $S$. Then
$T$ is the sink component of the minimum vertex $r$-cut. It has a
constant number of vertices, $k$, while the size of the vertex cut,
$n$, is arbitrarily large.
\begin{center}
  \includegraphics{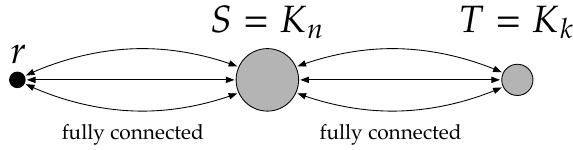}
\end{center}

That said, we show that a useful sparsification is possible if we
relax to \emph{approximating} the rooted vertex connectivity, and
qualify the lemma by the assumption that no singleton cut already
represents a good approximation. To this end, let $u,v \in V$. We say
that $u$ is an \defterm{in-neighbor} of $v$ if $(u,v) \in E$. We
denote the set of in-neighbors of a vertex $v$ by
\begin{math}
  \inneighbors{v}               %
  \defeq                             %
  \setof{u \in V \where (u,v) \in E}.
\end{math}
The definition of in-neighbors naturally extends
to sets of vertices; for $S\subset V$ we define
$\inneighbors{S} = (\cup_{v \in S} \inneighbors{v}) \setminus S$.
The \defterm{weighted in-degree} of $v$ is defined as the total weight
of all in-neighbors of $v$. Similarly we define
the set of \defterm{out-neighbors} of a vertex $v$, denoted $\outneighbors{v}$, as
\begin{math}
  \outneighbors{v}
  \defeq \setof{u \in V \where (u,v) \in E},
\end{math}
and the weighted out-degree of $v$, denoted $\outdegree{v}$, as the
sum of weights over $\outneighbors{v}$.

Our first lemma gives an approximate relationship between the weight
of the minimum weight rooted vertex cut and the weight of the sink
component of the minimum weight rooted vertex cut.

\begin{restatable}{lemma}{RootedSmallComponentSmallVertexCut}
  Let $\eps > 0$ be fixed.  \labellemma{small-sink-small-vertex-cut}
  \setupGrvc Suppose the in-neighborhood of every non-root vertex has total
  weight greater tha $\epsmore \vc$.  Then the minimum vertex $r$-cut
  has more than $\eps \vc$ weight in the sink component.
\end{restatable}

\begin{proof}
  Let the minimum $r$-cut be of the form $\inneighbors{S}$, where
  $S \subseteq V - r$. To prove the claim it suffices to show that
  $w(S) > \eps \vc$.

  For any vertex $v \in S$, by assumption, total weight of
  in-neighbors is more tha $\epsmore \vc$. At most $\vc$ weight of
  these in-neighbors are in the minimum vertex $r$-cut,
  $\inneighbors{S}$. This implies that $v$ has more than $\eps \vc$
  weight of in-neighors in $S$, and hence
  $\sum_{s \in S} \weight{s} > \eps \vc + 1$ (where the extra $1$ is
  for the weight of $v$).
\end{proof}

\begin{figure}
  \centering                    %
  \includegraphics{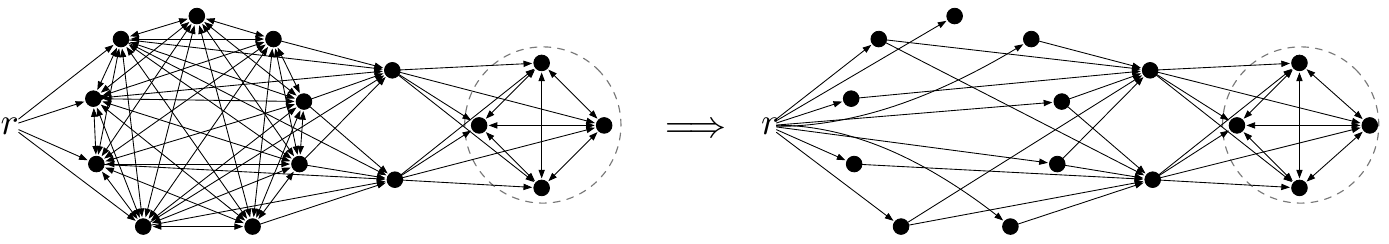} %
  \caption{An example of the Rooted Sparsification Lemma for vertex
    connectivity in action. In the input graph on the left, minimum
    vertex $r$-cut has size 2 and the sink component (circled) has 4
    vertices. The minimum in-degree (other than $r$) is 5. On the
    right hand side, all vertices with in-degree $\geq 9$ have all
    their incoming edges replaced with a single edge from $r$. The
    minimum vertex $r$-cut is again 2 and the sink-component of the
    minimum $r$-cut remains unchanged.}

  \graybar
\end{figure}

\begin{restatable}{lemma}{ApxSimpleVertexSparsification}
  \labellemma{vertex-sparsification} Let $\eps > 0$ be
  fixed. \setupGrvc \setupInDegree Let $k \in
  \naturalnumbers$. Consider the graph $\bar{G}$ obtained by
  replacing, for each vertex $v \in V$ with weighted in-degree
  $\geq \parof{1 + 1/\eps} k$, all of the in-coming edges to $v$ with
  a single edge from $r$ to $v$. Then we have the following.
  \begin{enumerate}
  \item $\bar{G}$ has maximum weighted in-degree at most $(1 + 1/\eps)
    nk$.
  \item $\bar{G}$ has at most $\parof{1+1/\eps} nk$ edges.
  \item If the sink component of a minimum vertex $r$-cut in $G$ has
    weight $\leq k$, then the minimum vertex $r$-cut in $G$ and
    $\bar{G}$ are the same.
  \end{enumerate}
\end{restatable}

\begin{proof}
  Let $T$ be the sink component of a minimum rooted $r$-vertex cut, of
  minimum weight among such sink components. Suppose $T$ has weight
  less than or equal to $k$.  By
  \reflemma{small-sink-small-vertex-cut}, $\vc < k / \eps$. Therefore
  any vertex $v$ with weighted in-degree greater than
  $\parof{1 + 1 / \eps} k$ cannot be in $T$ of the minimum rooted
  $r$-vertex cut.  We claim that replacing the incoming edges to $v$
  does not decrease rooted vertex connectivities for $r$. As a thought
  experiment, suppose we make the replacement over two steps, where we
  first add the edge from $r$ to $v$, and then remove the other
  incoming edges to $v$. The first step does not decrease vertex
  connectivities, and forces the rooted vertex connectivity from $r$
  to $v$ to be $+\infty$. Removing the other incoming edges to $r$
  does not effect the connectivity from $r$ to $v$, so no other vertex
  connectivities from $r$ are effected either. Over the two steps,
  then, we see that no rooted connectivities from $r$ decrease.

  On the other hand, since $v$ is not in the sink of the minimum
  vertex $r$-cut, the rooted vertex connectivity of $r$ does not
  change.
\end{proof}

\subsection{Rooted vertex connectivity parametrized by sink component size}

\labelsection{parametrized-vertex-cut}

We now give an algorithm for rooted vertex connectivity parametrized
by the weight of the vertices in the sink component of the minimum rooted
cut. More precisely, we take as additional input two weights
$k_1 \leq k_2$ and assume the sink component has weight between
$k_1$ and $k_2$.

In the following, let $\vctime{m}{n}$ be the running time for vertex
$(s,t)$-cut.

\begin{lemma}
  \labellemma{fixed-sink-size-vertex-cut} Let $\eps > 0$ be fixed.
  \setupGrvc Let $W = \sum_{v \in V} w(v)$ be the total weight in the
  graph. \setupInDegree Let $k_1,k_2 \in \naturalnumbers$ with
  $0 < k_1 < k_2$.  Suppose also that the sink component of the
  minimum $r$-cut has between $k_1$ and $k_2$ total weight. Then the
  minimum $r$-cut can be computed with constant probability in
  \begin{align*}
    \bigO{m + \prac{W - \vc}{k_1} \vctime{\min{m, \frac{k_2 n}{\eps}}, n}}
  \end{align*}
  randomized time.
\end{lemma}

\begin{proof}
  By \reflemma{vertex-sparsification}, in $\bigO{m}$ time, we can
  reduce the number of edges to at most $\bigO{k_2 n / \eps}$ without
  decreasing the rooted vertex connectivity. We sample vertices from
  $V' = V \setminus \parof{\setof{r} \cup \outneighbors{r}}$. Note
  that $V'$ has weight at most $W - \outdegree{r} \leq W - \vc$.

  We sample
  $\bigO{(W - \outdegree{r}) /k_1} \leq \bigO{(W - \vc) / k_1}$
  vertices
  \begin{math}
    t \in V'
  \end{math}
  independently in proportion to their weight.  For each sampled
  vertex $t$, we compute the minimum $(r,t)$-vertex cut in the
  sparsified graph. With constant probability, one of these vertices
  $t$ is in the sink component of the minimum vertex $r$-cut, and the
  minimum vertex $(r,t)$-cut is the minimum vertex $r$-cut.
\end{proof}
\begin{remark}
  The simple observation that the weight of $\outneighbors{r}$ is at
  least $\vc$ is from \cite{ghr}.
\end{remark}

\subsection{Rooted vertex connectivity for small sink
  components}

\providecommand{\involume}{\fparnew{\operatorname{vol}^{-}}}
\providecommand{\splitG}{G_{\text{s}}}
\providecommand{\splitT}{T_{\text{s}}}
\providecommand{\vin}{v^-}%
\providecommand{\vout}{v^+}%
\providecommand{\uin}{u^-}%
\providecommand{\uout}{u^+}     %
\providecommand{\ecB}{\ec^*}    %
\providecommand{\ellB}{\ell^*}    %
\providecommand{\revG}{G_{\text{rev}}}

This section develops an approximation algorithm for rooted vertex
connectivity specifically for the case where the sink component has
small weight. The algorithm takes an upper bound $k$ on the weight of
the sink component, and guarantees an approximate minimum cut when
there is a minimum rooted vertex cut where the sink component has
weight at most $k$. The approach is inspired by the recent local
connectivity algorithm of \cite{fnsyy}, and also integrates the rooted
sparsification lemma. This algorithm is developed in two steps. The
first step is a local cut algorithm that, given a vertex $t \in V$,
searches for a small $(r,t)$-cut around $t$ in time proportional to a
given upper bound on the weight of the sink component. The second step
first applies the rooted sparsification lemma, finds a vertex $t$ in
the sink component by random sampling, and runs the local cut
algorithm for this choice of $t$.

The following lemma, which describes the local cut algorithm, is
nearly identical to \cite{fnsyy} except for two small
modifications. First, we work with integral capacities, which does not
change any arguments. Second is the inclusion of the root $r$ which we
want to keep on the opposite side of the local cut. The proof is
included for the sake of completeness. In the following, the
\defterm{in-volume} of a set of vertices $T$ in a directed,
edge-capacitated graph is the sum of weighted in-degrees over all
vertices in $T$. Similarly the out-volume is defined as the sum of
weighted out-degrees.

\begin{lemma}
  \labellemma{rooted-local-cut} Let $G = (V,E)$ be a directed graph
  with integral edge capacities. We assume that $G$ is already
  available in memory in adjacency list format.  Let $r, t \in V$, and
  let $\ec, \ell > 0$ be given parameters. Then there is a randomized
  algorithm that runs in $\bigO{\ell \ec/\eps}$ time with the
  following guarantee.

  Let $\ecB$ be the minimum capacity of all $(r,t)$ cuts where the
  sink component has in-volume at most $\ell$. If $\ecB < \ec$, then
  with constant probability, the algorithm returns an $(r,t)$-cut of
  capacity at most $\epsmore \ecB$.
\end{lemma}

\begin{proof}
  Let $T$ be the sink component of a minimum $(r,t)$-edge cut among
  those where the sink component has in-volume at most $\ell$.  We run
  a randomized variation of augmenting paths in the reversed graph
  $\revG$ where $t$ is the source. Note that $T$ now has out-volume at
  most $\ell$ in $\revG$. We run the following subroutine for at most
  $1 + \epsmore \ec$ iterations, where each iteration routes one unit
  of flow from $t$ to some chosen $v$.

  Each iteration $i$ runs DFS from $t$ in the residual graph, until it
  either (a) visits $r$, (b) has traversed edges of total capacity at
  least $\bigO{\ell / \eps}$, or (c) has explored all the edges
  reachable from $t$ while failing to satisfy either (a) or (b). In
  event (a), we route one unit of flow to $r$. In event (b), we select
  one of the visited edges randomly in proportion to their capacity,
  and route one unit of flow to the endpoint of that edge. In either
  case, after routing, we update the residual graph by reverse (one
  unit of capacity) of each edge on the path from $t$ to the selected
  sink.  In event (c), we return the entire component of vertices
  reachable from $t$ which induces an $(r,t)$-cut in the original
  graph.  If, after $1 + \epsmore \ec$ iterations, we never reach
  event (c), then the algorithm terminates with failure.

  We first argue that we return a $\epsmore$-approximate $(r,t)$-cut
  with constant probability.  We first point out that the total
  out-volume of $T$ in the residual graph never increases, as we are
  reversing edges along edges along a path starting from $t$. Next, we
  observe that in each instance of event (b), where we randomly sample
  the endpoint of a visited edge as a sink, there is less than
  $\eps / 2$ probability that this endpoint lies in $T$. This is
  because the graph search has traversed a total capacity of at least
  $\bigO{\ell / \eps}$, and $T$ has out-volume at most $\ell$. That
  is, the out-volume of $T$ represents at most an
  $(\eps / 2)$-fraction of the searched edges.

  Now, over the first $\epsmore \ecB$ iterations, we expect to sample
  less than $\eps \ecB/ 2$ sinks from $T$. By Markov's inequality, we
  sample less than $\eps \ecB$ sinks from $T$ over the first
  $\epsmore \ecB$ iterations with probability at least $1/2$. In this
  event, if the algorithm did not find an $(r,t)$-cut within the first
  $\ecB$ iterations, then we must have routed more than $\ecB$ units
  of flow out of $T$ -- a contradiction. Thus the algorithm finds an
  $(r,t)$-cut within $\epsmore \ecB$ iterations with probability at
  least $1/2$. Since this cut was obtained as the reachable set of $t$
  after routing at most $\epsmore \ecB$ units of flow, the cut has
  capacity $\leq \epsmore \ecB$.

  It remains to prove the running time. Each iteration takes
  $\bigO{\ell / \eps}$ time to traverse at most $\bigO{\ell / \eps}$
  edges. The algorithm runs for at most $\bigO{\ec}$ iterations.
\end{proof}

The next lemma presents the approximate rooted vertex cut algorithm
that uses \reflemma{rooted-local-cut} as a subroutine. It also uses
the rooted sparsification lemma to reduce the size of the graph and
give stronger bounds on the volume of the sink component of the
desired vertex cut.

\begin{lemma}
  \labellemma{vc-small-sink} Let $\eps > 0$ be fixed. \setupVGintr Let
  $k \in \naturalnumbers$ and suppose that the sink component of the
  minimum $r$-cut has weight $\leq k$. Then a $\epsmore$-approximate
  minimum $r$-cut can be computed with high probability in
  $\bigO{m + (W - \vc) k^2 \log{n} \log{k} / \eps^3}$ randomized time.
\end{lemma}

\begin{proof}
  By \reflemma{small-sink-small-vertex-cut}, either a
  $\epsmore$-approximate minimum cut is induced by a singleton, or the
  minimum $r$-vertex cut has weight at most $\bigO{k / \eps}$. The
  former is addressed by inspecting all singleton cuts. For the rest
  of the proof, let us assume the latter. By
  \reflemma{vertex-sparsification}, we can sparsify the graph to have
  maximum weighted in-degree $\bigO{k / \eps}$, hence at most
  $\bigO{n k / \eps}$ total edges.

  Let $T$ be the sink component of the minimum $r$-cut, which has
  total vertex weight at most $\bigO{k}$, and induces an $r$-cut with
  capacity $\vc \leq \bigO{k/\eps}$. Recall the standard auxiliary
  ``split-graph'' where vertex capacities are modeled as edge
  capacities.  The high-level idea is to find a vertex $t \in T$ by
  random sampling and then apply \reflemma{rooted-local-cut} to the
  appropriate auxiliary vertices of $r$ and $t$ in the split graph.

  To this end, we first bound the volume of the sink-component
  corresponding to $T$ in the split-graph. We recall that the split
  graph splits each vertex $v$ into an auxiliary ``in-vertex'' $\vin$
  and an auxiliary ``out-vertex'' $\vout$. For each $v$ there is a new
  edge $(\vin,\vout)$ with capacity equal to the vertex capacity of
  $v$.  Each edge $(u,v)$ is replaced with an edge $(\uout,\vin)$ with
  capacity\footnote{Usually, this edge is set to capacity $\infty$,
    but either the weight of $u$ or the weight of $v$ are also valid.}
  equal to the vertex capacity of $u$.  As a sink component, $T$ maps
  to a vertex set $T'$ in the split-graph consisting of (a) both
  copies $\vin$ and $\vout$ of each vertex $v \in T$, and (b) the
  out-vertex $\vout$ of each vertex $v$ in the vertex in-cut
  $\inneighbors{T}$. For each vertex $v \in T$, $\vin$ has
  (edge-)weighted in-degree equal to the vertex-weighted in-degree of
  $v$ in the original graph, which is at most $\bigO{k / \eps}$. This
  sums to $\bigO{\sizeof{T} k / \eps}$ over all $v \in T$. For each
  $v \in T$, $\vout$ has weighted in-degree equal to the vertex weight
  of $T$, which sums to the total vertex weight of $T$. Lastly, for
  each $v \in \inneighbors{T}$, $\vout$ has weighted in-degree equal
  to the vertex weight of $v$. This sums to $\vc \leq \bigO{k / \eps}$
  over all $v \in \inneighbors{T}$. All summed up, the total in-volume
  of $T'$ in the split-graph is at most $\bigO{k / \eps}$ times the
  total vertex weight of $T$.

  Suppose we had a constant factor estimate $\ell$ for the total
  vertex weight of $T$. Then we can sample
  $\bigO{(W - \outdegree{r}) \log{n} / \ell} \leq \bigO{(W - \vc)
    \log{n} / \ell}$ vertices by weight from
  $V \setminus \parof{\setof{r} \cup \outneighbors{r}}$. With high
  probability, we sample $\bigO{\log n}$ vertices from $T$. For each
  sampled vertex $t$ we invoke \reflemma{rooted-local-cut} to find an
  $(r,t)$-cut, with upper bound $\bigO{\ell k / \eps}$ on the volume
  of the sink component and $\bigO{k/ \eps}$ as the upper bound on the
  cut. With high probability, one of these calls returns a
  $\epsmore$-approximate cut. The total time, over all calls, would be
  $\bigO{(W - \vc) \log{n} k^2 / \eps^3}$.

  Of course, we do not know the vertex weight of $T$ \emph{a
    priori}. However, we know that it is upper bounded by $k$, and let
  $\ell$ enumerate all powers of 2 between $1$ and $k$. For each
  $\ell$, run the process described above under the hypothesis that
  $\ell$ is a constant factor estimate for the total vertex weight of
  $T$. Each choice of $\ell$ takes
  $\bigO{(W - \vc) \log{n} k^2 / \eps^3}$ time. There are
  $\bigO{\log{k}}$ choices of $\ell$. One of these choices of $\ell$
  is a constant factor for the total volume of $T$ and produces a
  $\epsmore$-approximate minimum $(r,t)$-cut with high probability.
\end{proof}

\subsection{Rooted vertex connectivity}

\labelsection{rooted-vc}

We now present the algorithm for approximate rooted vertex
connectivity and prove \reftheorem{simple-vc-rooted}. The algorithm
combines the subroutine in \reflemma{fixed-sink-size-vertex-cut} for
logarithmically many ranges of weights, and \reflemma{vc-small-sink}
for sufficiently small weights. We restate
\reftheorem{simple-vc-rooted} for the sake of convenience.

\SimpleVCRooted*

\begin{proof}
  Let $\vc_0 = \eps \sqrt{n}$. Let $i_0 = \logdown{\vc_0}$, and let
  $i_1 = \max{\logup{\eps m / n}, i_0 + 1}$ For each
  $i = \logdown{\vc_0}, \logdown{\vc_0} + 1, \dots, i_1-1$, let
  $k_i = 2^{i}$.  Let $k_{i_1} = W - \outdegree{r}$ where we recall
  that $\outdegree{r}$ is the weighted out-degree of $r$.  For each
  $i$, we apply \reflemma{fixed-sink-size-vertex-cut} with lower bound
  $k_i$ and upper bound $k_{i+1}$ on the weight of the sink component
  of the minimum vertex $r$-cut. We repeat this subroutine
  $\bigO{\log n}$ times for each $i$ to amplify the success
  probability from constant to high probability. We use
  $\vctime{m}{n} = \apxO{m + n^{1.5}}$ \cite{brand+}. We also apply
  \reflemma{vc-small-sink} with $\eps \vc_0$ has an upper bound on the
  sink component size. The set of all cuts obtained by these methods
  includes a $\epsmore$-approximate minimum $r$-cut with high
  probability, and we return the minimum of these cuts.  The combined
  running time is
  \begin{align*}
    \apxO{m + \frac{(W - \vc) n}{\eps} + \frac{(W - \vc) n^{1.5}}{
    \vc_0} + (W - \vc) \vc_0^2 / \eps^3} %
    \leq                           %
    \apxO{m + (W - \vc) n / {\eps}},
  \end{align*}
  as desired. The exact bound follows by first using the approximation
  algorithm to obtain a constant factor estimate for $\vc$, and then setting
  setting $\eps \leq 1/\vc$.
\end{proof}

\subsection{Global vertex connectivity}

\labelsection{global-vertex-connectivity}

We now shift to global vertex connectivity and prove
\refcorollary{simple-vc-global}, which we address by reduction to the
algorithm for rooted vertex connectivity above. We note that obtaining
a root is slightly non-trivial because many vertices may be in the
minimum weight vertex cut. We restate \refcorollary{simple-vc-global}
for the sake of convenience.

\SimpleVCGlobal*
\begin{proof}
  Let $\vc$ denote the global vertex connectivity.  If we sample a
  single vertex $r$ in proportion to its weight, then with probability
  $1 - \vc / W$, $r$ is not in the minimum vertex cut. Then either the
  rooted vertex connectivity from $r$, or to $r$ (i.e., from $r$ in
  the graph $G'$ with all the edges reversed), will give the rooted
  vertex cut. In principle we would like to apply
  \reftheorem{rooted-vc-reduction} with root $r$ in both orientations,
  which conditional on $r$ not being in the minimum cut, succeeds with
  high probability. We amplify by repeating
  $L = O(\frac{W}{W - \kappa}\log n)$ times to obtain the high
  probability bound. Observe that the running time, via
  \reftheorem{rooted-vc-reduction}, is
  \begin{align*}
    \apxO{m L + n W / \eps}.
  \end{align*}
  We would like to remove the $m L$ factor.

  To this end, observe that the $m$ term arises from applying the
  rooted sparsification lemma for various estimates $k$ of the weight
  of the sink component.  Recall that for fixed $k$ and $\eps$, the
  sparsification lemma replaces, for every vertex $v$ with in-degree
  $> \bigO{k / \eps}$, all the incoming edges to $v$ with a single
  edge from the root. Note that much of the sparsification lemma can
  be executed without $r$.  In particular, we can remove all incoming
  edges to the high in-degree vertices without knowing $r$; once $r$
  is given, we add an edge from $r$ to each of these vertices. The key
  point is that the first part, which takes $\bigO{m}$ time, can be
  done once for all $L$ sampled roots for each value of
  $k$. Thereafter, each of the $L$ roots takes $\bigO{n}$ to complete
  the sparsification for that root. This replaces the $\apxO{m L}$
  term with $\apxO{n L}$, which is dominated by $\apxO{n W / \eps}$.

  For the exact algorithm, we first apply the approximation algorithm
  with $\eps = 1/2$ obtain a factor-2 approximation to $\vc$ within
  the claimed running time. We then apply the approximation algorithm
  again with $1 / (2 \vc) \leq \eps \leq 1/\vc$.
\end{proof}
